%% file: mselection.tex
\documentclass[a4paper,UKenglish,cleveref, autoref]{lipics-v2019}

\nolinenumbers 

\usepackage{amsmath,amssymb} 
\usepackage{tikz}
\usepackage{enumitem}
\usepackage{mathtools}
\usepackage{verbatim}

\newtheorem{notation}{Notation}

\input{tikzstyles.tikzstyles}

\renewcommand{\P}{\mathcal{P}_{\mathrm{f}}} 
\newcommand{\pf}{\P} 
\newcommand{\rat}{\mathrm{Rat}} 
\renewcommand{\sp}{\mathrm{SP}} 
\newcommand{\R}{\mathbb{R}} 
\newcommand{\G}{\mathcal{G}} 
\newcommand{\eps}{\epsilon} 
\newcommand{\J}{\mathcal{J}} 
\DeclareMathOperator*\argmax{\arg\max}
\newcommand\bind{>\!\!\!>\!\!=} 

\let\epsilon\varepsilon

\usepackage{xcolor} 

\usepackage{hyperref} 

\title{Sequential games and nondeterministic selection functions}
\author{Joe Bolt}{Department of Computer Science, University of Oxford}{}{}{}
\author{Jules Hedges}{Department of Computer Science, University of Oxford}{}{}{}
\author{Philipp Zahn}{Department of Economics, University of St. Gallen}{}{}{}

\authorrunning{J. Bolt, J. Hedges, and P. Zahn}

\Copyright{Joe Bolt, Jules Hedges, and Philipp Zahn}



\begin{document}

\maketitle

\begin{abstract}
This paper analyses Escard\'o and Oliva's generalisation of selection
functions over a strong monad from a game-theoretic perspective. We focus on the case of the nondeterminism (finite nonempty
powerset) monad $\P$. We use these nondeterministic selection
functions of type $\J^{\P}_R X = (X\rightarrow
R)\rightarrow \P(X)$ to study sequential games, extending previous
work linking (deterministic) selection functions to game theory. 
Similar to deterministic selection functions, which compute a subgame perfect
Nash equilibrium play of a game, we characterise those non-deterministic
selection functions which have a clear game-theoretic interpretation. We show, surprisingly, no non-deterministic selection function exists which computes the set of all subgame perfect Nash equilibrium plays. Instead we show that there are selection functions corresponding to sequential versions of the iterated removal of strictly dominated strategies.
\end{abstract}

\section{Introduction}

Selection functions are an approach to games of perfect information developed by Escard\'o and Oliva in a series of papers beginning with \cite{escardo_oliva_selection_functions_bar_recursion_backward_induction}.
As well as revealing a deep connection between game theory and proof theory,
this approach elucidates the mathematical structure of backward induction, a
method to compute equilibria, showing that it arises from a more primitive algebraic structure known as the selection monad.
A more general form of the selection monad was developed in \cite{escardo_oliva_herbrand_interpretation} for proof-theoretic purposes, with a special case appearing in \cite{hedges_monad_transformers_backtracking_search}.
In this paper we explore this more general structure from a game-theoretic perspective.

A selection function is a function (specifically a type-2 function, also known as a functional or operator, i.e. a function whose domain is a set of functions) of type $(X \to R) \to X$, which we write $\J_R X$.
The operator $\J_R$, which associates to every set $X$ the set of functions $(X \to R) \to X$, is called the \emph{selection monad}, and it carries an algebraic structure known as a \emph{strong monad}.
One consequence of this structure is that there is a family of product-like operators
\[ \bigotimes : \prod_{i = 1}^n \J_R X_i \to \J_R \prod_{i = 1}^n X_i \]
called the \emph{product of selection functions}.

An important class of selection functions are those of the form $\epsilon : \J_{\R^n} X$ satisfying $\epsilon (k) \in \argmax (\pi_i \circ k)$ for all $k : X \to \R^n$, where
\[ \argmax (f) = \{ x : X \mid f (x) \geq f (x') \text{ for all } x' : X \} \]
Consider an $n$-player game of perfect information. That is to say, all players move in turn and the current state of the game is known all players. In round $i$, player $i$ selects a move from the set $X_i$, with payoffs for all players given by $q : \prod_{i = 1}^n X_i \to \R^n$.
The usual solution concept for games of this form is known as subgame perfect
equilibrium -- a strengthening of Nash equilibrium. These equilibria can be computed using a method known as backward induction which dates back at least to Zermelo \cite{schwalbe_walker_zermelo_early_history_game_theory}.

The key result connecting selection functions with game theory is that if $\epsilon_i$ is a sequence of selection functions $(1 \leq i \leq n)$ satisfying $\epsilon_i (k) \in \argmax (\pi_i \circ k)$ for all $k : X_i \to \R^n$, then
\[ \left( \bigotimes_{i = 1}^n \varepsilon_i \right) (q) \]
is the strategic play of some subgame perfect equilibrium (i.e. the play that results when the players' strategies form a subgame perfect equilibrium).
The $\bigotimes$ operator expresses the essence of the backward induction method in a mathematically pure way.
Surprisingly, this result extends to $n = \infty$ when $q$ is continuous, a
highly non-obvious fact about backward induction that is often directly
contradicted in game theory textbooks. For instance, \cite[pp. 107-
108]{fudenberg1991game}:

\begin{quote}
When the horizon [of the sequential game] is infinite, the set of
subgame-perfect equilibria cannot be determined by backward induction from the
terminal date, as it can be in the finitely repeated prisoner's dilemma and in
any finite game of perfect information.
\end{quote}

Notice that the operator $\argmax$ itself has the type $(X \to \R) \to \P (X)$,
which we write $\J^{\P}_\R X$, when $X$ is finite and nonempty. Also note that in this paper $\P$ refers to the finite \emph{nonempty} powerset monad.
In \cite{escardo_oliva_herbrand_interpretation} it is proved that $\J^{\P}_R$ is a strong monad when $R$ is a meet-semilattice (and more generally that $\J^T_R$ is a strong monad when $T$ is a strong monad and $R$ is a $T$-algebra), however this has so far not been addressed through the lens of game theory.
It follows that each choice of meet-semilattice structure on $\R$ induces a product
\[ \bigotimes_{i = 1}^n \argmax : \J_\R \prod_{i = 1}^n X_i \]
and hence, if $q$ is the payoff function of a game, we obtain a set of plays $\left( \bigotimes_{i = 1}^n \argmax \right) (q)$.
One may conjecture that there is some choice of meet-semilattice
structure for which this set is the set of plays of \emph{all} subgame perfect
equilibria.

We prove that this is not the case.
We complement this negative characterization by a positive one: If one replaces the Nash subgame condition
by a weaker condition, which we call \emph{rational}, then the product of
multivalued selection functions does exactly
characterize the set of all plays in line with this condition. Examples of this
weaker condition are weak and strict dominance of strategies -- standard
solution concepts in game theory.

\textbf{Outline.} In section 2.1 we give background on selection functions, and in section 2.2 on higher-order sequential games. In section 3 we define a condition on selection functions, and characterise for 2-round games the product of those selection functions that satisfy this condition. In section 4 we show that the product of selection functions does not compute subgame perfect equilibria in general, and also characterise those cases when it does. In section 5 we give a more general theorem characterising the product of selection functions for $n$-round games, and give concrete examples with selection functions that pick out strictly dominated strategies.

Additional proofs can be found in the appendix.

\section{Preliminaries}

We begin with introducing selection functions and sequential games of perfect information. 

\subsection{Selection functions}

Formally, in this paper we work over some fixed cartesian closed category which contains analogues of finite sets and $\R$.
Since we are working only with finite games, there is no harm in taking this to be simply the category of sets.
We will treat monads `Haskell-style', defined by their action on objects, their unit and their bind operator (which is Kleisli extension with the arguments swapped).

\begin{definition}
	Let $T$ be a strong monad and $\alpha : T R \to R$.
	The $T$-selection monad is defined by $\J^T_R X = (X \to R) \to T X$ with the following monad operations:
	\begin{itemize}
		\item The unit $\eta^{\J^T_R}_X : X \to \J^T_R X$ is defined by $\eta^{\J^T_R}_X (x) = \lambda (k : X \to R) . \eta^T_X (x)$
		\item The bind operator $\bind^{\J^T_R} : \J^T_R X \times (X \to \J^T_R Y) \to \J^T_R Y$ is defined by
		\[ \varepsilon >\!\!\!>\!\!=^{\J^T_R} f = \lambda (k : Y \to R) . \varepsilon (h) \bind^T g \]
		where $g : X \to T Y$ is defined by $g (x) = f (x) (k)$ and $h : X \to R$ is defined by $h (x) = \alpha \left( g (x) \bind^T (\eta^T_R \circ k) \right)$.
	\end{itemize}
\end{definition}

\begin{proposition}[\cite{escardo_oliva_herbrand_interpretation}, Lemma 2.2]
	If $\alpha : T R \to R$ is a $T$-algebra then $\J^T_R$ is a strong monad.
\end{proposition}

As a special case, when $T$ is the identity monad we use the notation $\J_R X = (X \to R) \to X$, which has bind operator
\[ \varepsilon \bind^{\J_R} f = \lambda (k : Y \to R) . g (\varepsilon (k \circ g)) \]
where $g : X \to Y$ is defined by $g (x) = f (x) (k)$.
Since the identity function uniquely makes every type into an algebra of the identity monad, $\J_R$ is a strong monad for every $R$.

Every strong monad $T$ admits a \emph{monoidal product} operator
\[ \otimes : T X \times T Y \to T (X \times Y) \]
(in fact it admits two in general, and we take the `left-leaning' one), and a more general \emph{dependent monoidal product} operator
\[ \otimes : T X \times (X \to T Y) \to T (X \times Y) \]
For example, for the finite nonempty powerset monad $\P$ the monoidal product is given by cartesian product of sets
\[ a \otimes^{\P} b = \{ (x, y) \mid x \in a, y \in b \} \]
and the dependent monoidal product is the `dependent cartesian product'
\[ a \otimes^{\P} f = \{ (x, y) \mid x \in a, y \in f (x) \} \]

For the purposes of this paper we will only need the simple monoidal product of the selection monad, but it is defined in terms of the dependent monoidal product of the underlying monad.
Concretely, the simple monoidal product
\[ \otimes^{\J^T_R} : \J^T_R X \times \J^T_R Y \to \J^T_R (X \times Y) \]
is defined by
\[ \varepsilon \otimes^{\J^T_R} \delta = \lambda (k : X \times Y \to R) . a \otimes^T f \]
where $a : T X$ is defined by $a = \varepsilon \left( \lambda x . \alpha \left( f (x) \bind^T \left( \lambda y . \eta^T_R (k (x, y)) \right) \right) \right)$ and $f : X \to T Y$ is defined by $f (x) = \delta (\lambda y . k (x, y))$.

When $T$ is the identity monad this simplifies to
\[ \varepsilon \otimes \delta = \lambda (k : X \times Y \to R) . (a, f (a)) \]
where $a = \varepsilon (\lambda x . q (x, f x))$ and $f (x) = \delta (\lambda y . q (x, y))$.
Alternatively, when $T = \P$ it is
\[ \varepsilon \otimes \delta = \lambda (k : X \times Y \to R) . \{ (x, y) \mid x \in a, y \in f (x) \} \]
where $a = \varepsilon \left( \lambda x . \bigvee \{ q (x, y) \mid y \in f (x) \} \right)$ and $f (x) = \delta (\lambda y . q (x, y))$, where the $\P$-algebra is written $\bigvee : \P R \to R$.

\subsection{Higher-order sequential games}

A sequential game of perfect information is one in which players take turns sequentially, one player per round, with each player being able to perfectly observe the moves made in earlier rounds.

Higher-order sequential games are a generalisation introduced in \cite{escardo_oliva_selection_functions_bar_recursion_backward_induction} in which each player carries a selection function that defines what they consider `rational'. Ordinary (or `classical') sequential games are obtained as the special case in which every player's selection function is $\argmax$. An in-depth discussion of the decision-theoretic and game-theoretic content of higher-order games can be found in \cite{hedges_etal_higher_order_decision_theory,hedges_etal_selection_equilibria_higher_order_games}.

\begin{definition}
	An $n$-round higher order sequential game of perfect information is defined by the following data:
	\begin{itemize}
		\item For each player $1 \leq i \leq n$, a finite nonempty set $X_i$ of choices
		\item A set $R$ of outcomes, and an outcome function $q : \prod_{i = 1}^n X_i \to R$
		\item For each player $1 \leq i \leq n$, a multi-valued selection function $\epsilon_i : \J_R^{\P} X_i$
	\end{itemize}
\end{definition}

There are several small variants of this definition in the literature, which replace multi-valued selection functions with related higher order functions.
The original definition in
\cite{escardo_oliva_selection_functions_bar_recursion_backward_induction}
equipped players with a `quantifier' of type $(X \to R) \to R$ rather than a
selection function, with the motivating example being $\max : (X \to \R) \to
\R$. This was then generalised to multi-valued quantifiers of type $(X \to R)
\to \mathcal P (R)$ in \cite{escardo11}. The definition stated above was given in
\cite[section 1.3]{hedges_towards_compositional_game_theory}, and is based on
the definition of Nash equilibrium for higher-order \emph{simultaneous} games in
\cite{hedges_etal_selection_equilibria_higher_order_games}. Here, we apply this
definition in the context of sequential games. 

We will often focus on 2-player sequential games, in which we write the sets of choices as $X$ and $Y$, and the selection functions as $\epsilon : \J_R^{\P} X$ and $\delta : \J_R^{\P} Y$.

\begin{definition}
	An $n$-round classical sequential game of perfect information is a sequential
  game in which the set of outcomes is $R = \R^n$, and the selection functions are
	\[ \epsilon_i (k) = \argmax (\pi_i \circ k) \]
\end{definition}

A classical game is one in which each player receives a real-valued outcome, and
players act such as to maximise their individual outcome, with no preference
over the outcomes of the other players. Note that we will refer to such outcomes
as payoffs, and to the outcome functions as payoff functions as it is the standard
in classical game theory. 
By different choices of $q$, this allows representation of both conflict and cooperative situations, and games with aspects of both, as is standard in game theory.

\begin{definition}
	A \emph{strategy} for player $i$ in a sequential game is a function $\sigma_i : \prod_{j = 1}^{i - 1} X_j \to X_i$ that makes a choice for player $i$ contingent on the choices observed in previous rounds.
	A \emph{strategy profile} is a tuple $\sigma : \prod_{i = 1}^n \left( \prod_{j = 1}^{i - 1} X_j \to X_i \right)$ consisting of a strategy for each player.
	The set of strategy profiles of a game is written $\Sigma$.
\end{definition}

Notice that a strategy for the first player is just an element of $X_1$, up to isomorphism.

\begin{definition}
	A \emph{play} of a sequential game is a tuple $x : \prod_{i = 1}^n X_i$ of choices.
	Every strategy profile $\sigma$ induces a play $\mathbf P (\sigma)$, called the \emph{strategic play} of $\sigma$, by `playing out', or more precisely by the course-of-values recursion
	\[ (\mathbf P (\sigma))_i = \sigma_i ((\mathbf P (\sigma))_1, \ldots, (\mathbf P (\sigma))_{i - 1}) \]
	This defines a function $\mathbf P : \Sigma \to \prod_{i = 1}^n X_i$. 
\end{definition}

This definition includes the base case $(\mathbf P (\sigma))_1 = \sigma_1 ()$, i.e. the choice made by the first player is just the choice that her strategy tells her to play.

\begin{definition}
A \emph{partial play} of a sequential game is a tuple $x :
\prod_{i=1}^j X_i$ for some $j<n$. 
Given a partial play $x$, a tuple of strategies
\[ \sigma : \prod_{i=j+1}^n \left( \prod_{k=1}^{i-1} X_k \rightarrow X_i \right) \]
induces a play $x^\sigma$ called the \emph{strategic extension} of $x$ by $\sigma$, given by
\[
x_i^\sigma = \begin{cases}
x_i & \text{if } i\leq j \\
\sigma_i(x_1^\sigma,\cdots,x_{i-1}^\sigma) & \text{otherwise.}
\end{cases}
\]
\end{definition}

A strategy profile of a 2-player sequential game is a tuple $\sigma : X \times (X \to Y)$, and the strategic play of $\sigma$ is $(\sigma_1, \sigma_2 (\sigma_1))$.
Also notice that the strategic extension of the empty partial play by $\sigma$ is $\mathbf P (\sigma)$.

Selection functions in  general describe what choices, plays or strategies are
`good' or `rational' for a player. But there are several possibilities we can
consider concretely. In the following, looking through the
lens of game theory, we investigate these candidates.

\section{Well-behaved selection functions}

We will now specify `niceness' constraints for multi-valued selection functions. Perhaps
unsurprisingly, the `niceness' of a multivalued selection function relates to its
interaction with the semilattice $R$. The correct definition of this property is
non-obvious, but does have a natural game theoretic interpretation.

\begin{definition}[Witnessing selection function]
	A multi-valued selection function $\epsilon:\J_R^{\pf}(X)$ is \textit{witnessing} if,
	for all indexing functions $I : X\rightarrow \pf(X\rightarrow R)$, if 
	\[
		x\in \epsilon \bigg(\lambda x'. \bigvee_{p\in Ix'} px' \bigg)
	\]
	then there exists a choice function $p_- : X\rightarrow (X\rightarrow R)$ for $I$
	(so $p_{x'}\in Ix'$ for all $x'$) such that 
	\[
		x \in \epsilon(\lambda x'. p_{x'}x').
	\]
\end{definition}

$Ix$ is thought of as the set of contexts that might arise if $x$ is chosen. The choice
function $p_-$ picks out a `plausible scenario,' a possible context for each choice that
could be made. In game theoretic terms, a witnessing selection function represents a
player that finds a move $x$ acceptable to play only if there is some plausible hypothesis
regarding how later players will behave under which $x$ is an acceptable move. Note that
the use of choice functions here does not invoke the axiom of choice as we are working
with finite nonempty powersets. 

\begin{definition}[Upwards closed selection function]
	A multi-valued selection function $\epsilon : \J_R^{\pf}(X)$ is \textit{upwards
	closed} if, whenever $p_- : X\rightarrow (X\rightarrow R)$ is a choice function for $I$ such
	that $x\in\epsilon (\lambda x'. p_{x'}x')$, it holds that
	\[
		x\in\epsilon\bigg( \lambda x'. \bigvee_{p\in Ix'} px' \bigg).
	\]
\end{definition}

Upwards closure is a converse notion to witnessing. If $x$ is an acceptable choice, then
$x$ remains an acceptable choice in contexts where other possible contexts are added and
then combined with the join operator (this notion is, admittedly, more game theoretically
vague but its interpretation will become clearer in the case where $R = \pf(\R)$ and the
semilattice join is given by union).

A good heuristic for thinking about witnessing and upwards closed selection functions is
as follows. Suppose $\eps : \J_R^{\pf}(X)$ and that $X = \{ x_1,\cdots, x_n\}$ and let $I:
X\rightarrow \pf (X\rightarrow R)$ with $Ix = \{ p^x_1\cdots, p^x_{m_x}\}$. We can
organise this information in a game tree:
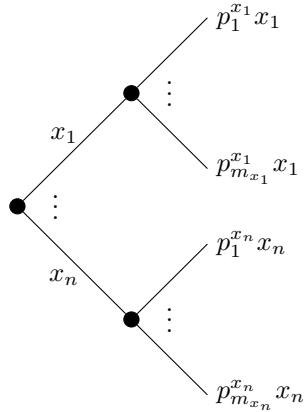
\begin{center}\begin{tikzpicture}
	\node (a)[varCopy]{};
	\path (a) +(1.5,1.5) node (a1) [varCopy]{};
	\path (a) +(1.5,-1.5) node (a2) [varCopy]{};
	\draw (a.45) -- node[above,pos=0.4,yshift=3pt]{$x_1$} (a1.225);
	\draw (a.-45) -- node[below,pos=0.4,yshift=-3pt]{$x_n$} (a2.135);

	\path (a1) +(1,1) node (o1)[anchor=west]{$p^{x_1}_1x_1$};
	\path (a1) +(1,-1) node (o2)[anchor=west]{$p^{x_1}_{m_{x_1}}x_1$};
	\draw (a1.north east) -- (o1.west);
	\draw (a1.south east) -- (o2.west);

	\path (a2) +(1,1) node (o3)[anchor=west]{$p^{x_n}_1x_n$};
	\path (a2) +(1,-1) node (o4)[anchor=west]{$p^{x_n}_{m_{x_n}}x_n$};
	\draw (a2.north east) -- (o3.west);
	\draw (a2.south east) -- (o4.west);

	\path (a) +(0.5,0.1) node{$\vdots$};
	\path (a1) +(0.5,0.1) node{$\vdots$};
	\path (a2) +(0.5,0.1) node{$\vdots$};
\end{tikzpicture}\end{center}
A choice function $q_- : X\rightarrow (X\rightarrow R)$ for $I$ then corresponds to
choosing a leaf of this tree for each $x\in X$. Visually (omitting dots for clarity), 

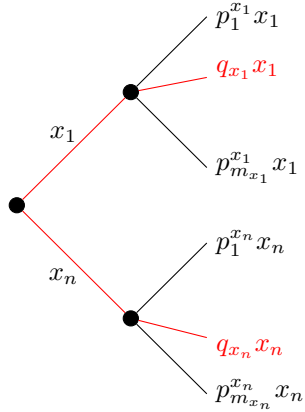
\begin{center}\begin{tikzpicture}
	\node (a)[varCopy]{};
	\path (a) +(1.5,1.5) node (a1) [varCopy]{};
	\path (a) +(1.5,-1.5) node (a2) [varCopy]{};
	\draw[red] (a.45) -- node[text=black,above,pos=0.4,yshift=3pt]{$x_1$} (a1.225);
	\draw[red] (a.-45) -- node[text=black,below,pos=0.4,yshift=-3pt]{$x_n$} (a2.135);

	\path (a1) +(1,1) node (o1)[anchor=west]{$p^{x_1}_1x_1$};
	\path (a1) +(1,-1) node (o2)[anchor=west]{$p^{x_1}_{m_{x_1}}x_1$};
	\draw (a1.north east) -- (o1.west);
	\draw (a1.south east) -- (o2.west);

	\path (a2) +(1,1) node (o3)[anchor=west]{$p^{x_n}_1x_n$};
	\path (a2) +(1,-1) node (o4)[anchor=west]{$p^{x_n}_{m_{x_n}}x_n$};
	\draw (a2.north east) -- (o3.west);
	\draw (a2.south east) -- (o4.west);

	\path (a1) +(1,0.3) node (c1)[anchor=west,text=red]{$q_{x_1}x_1$};
	\draw[red] (a1) -- (c1);

	\path (a2) +(1,-0.4) node(c2)[anchor=west,text=red]{$q_{x_n}x_n$};
	\draw[red] (a2) -- (c2);
\end{tikzpicture}\end{center}
The red subtree then corresponds to the context $\lambda x. q_{x}x$. Contrastingly, the
context $\lambda x. \bigvee_{p\in Ix}px$ corresponds to the collapsed game tree
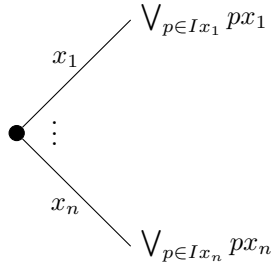
\begin{center}\begin{tikzpicture}
	\node (a)[varCopy]{};
	\path (a) +(1.5,1.5) node (o1)[anchor=west]{$\bigvee_{p\in Ix_1}px_1$};
	\path (a) +(1.5,-1.5) node (o2)[anchor=west]{$\bigvee_{p\in Ix_n}px_n$};
	\draw (a.north east) -- node[above,pos=0.4,yshift=3pt]{$x_1$} (o1.west);
	\draw (a.south east) -- node[below,pos=0.4,yshift=-3pt]{$x_n$} (o2.west);
	\path (a) +(0.5,0.1) node{$\vdots$};

\end{tikzpicture}\end{center}

A witnessing selection function is a selection function where, if $x$ is an acceptable
play in the collapsed tree, there is some choice of leaves such that $x$ is an acceptable
play in the associated context. An upwards closed selection function has the converse
property: if there is a choice of leaves under which $x$ is an acceptable play, then $x$
is an acceptable play in the collapsed tree.

\begin{example}
	We will show that $\arg\max$ is witnessing but not upwards closed. For a finite
	set $X$, define
	$\arg\max : (X\rightarrow\R)\rightarrow \pf (X)$ by
	\[
		\arg\max (k) = \Big\{ x\in X \Bigm| \forall x'\in X \quad kx \geq kx'
		\Big\}.
	\]
	$\arg\max$ is then a multi-valued selection function with the join operator on
	$\R$ given by $\max$. 

	\underline{Claim 1: $\arg\max$ is witnessing.} \textit{Proof:} Suppose \[x\in
	{\arg\max}\Big( \lambda x'. \max_{p\in Ix'} px' \Big)\]
	for some $I : X\rightarrow \pf(X\rightarrow \R)$. Then $\forall x'\in X$ it holds
	that
	\[
		\max_{p\in Ix} px \geq \max_{p\in Ix'} px'.
	\]
	As $Ix'$ is finite, we can choose $p_{x'}\in Ix'$ such that $p_{x'}x' = \max_{p\in
	Ix'}px'$. Then 
	\[
		x \in \arg\max \Big( \lambda x'. p_{x'}x' \Big).
	\]
	Hence $\arg\max$ is witnessing.
	
	\underline{Claim 2: $\arg\max$ is not upwards closed.} \textit{Proof:} Let
	$X=\{0,1\}$ and let $c_i : X\rightarrow \R$ denote the constant function $x\mapsto
	i$. Define $I : X\rightarrow \P(X\rightarrow \R)$ by
	\begin{align*}
		0 &\mapsto \{c_0\}
		\\
		1 &\mapsto \{c_1,c_{-1}\}
	\end{align*}
	Note that the function $\lambda x. \max_{p\in Ix} px$ is given by
	\begin{align*}
		0&\mapsto 0
		\\
		1&\mapsto 1
	\end{align*}
	and, hence, $\arg\max \Big(\lambda x. \max_{p\in Ix} px\Big) = \{ 1 \}$. Define a
	choice function $p_- : X\rightarrow (X\rightarrow \R)$ for $I$ by $p_0 = c_0$ and
	$p_1 = c_{-1}$. Then $\arg\max \Big( \lambda x. p_xx\Big) = \{0\}$, but $\{ 0
	\}\not\subseteq \{1\}$ and hence $\arg\max$ is not upwards closed.
\end{example}

That $\arg\max$ is witnessing follows from a more general result regarding
multi-valued selection functions for which the semilattice $R$ is total.

\begin{proposition}\label{proof1}
	If the semilattice $R$ is total, then for all sets $X$ and all selection functions
	$\epsilon : \J_R^{\pf}(X)$, $\epsilon$ is witnessing. 
\end{proposition}

We now consider an example of a multi-valued selection function which is upwards closed
but not witnessing.

\begin{example}
	Let $X=\{0,\star\}$. We think of $\star$ as $\epsilon$'s `favourite move'
	that they are happy to play in any context. Define a semilattice $R =
	\{\top,\bot_1,\bot_2\}$ where $\bot_1 \leq \top$, $\bot_2\leq \top$, and $\bot_1$
	and $\bot_2$ are not comparable. Define $\epsilon : \J_R^{\pf}(X)$ by 
	\[
		\epsilon(p) = \{\star\} \cup \{x\in X \:\:|\:\: px=\top\}.
	\]
	Suppose that $x\in \epsilon (\lambda x'.p_{x'}x')$ where $p_{x'}\in Ix'$ for some
	$I : X\rightarrow \pf(X\rightarrow R)$. Then either $x= \star$ or $p_x x = \top$.
	In either case, $x\in \epsilon \Big( \lambda x'. \bigvee_{p\in Ix'} px'\Big)$.
	Hence $\epsilon$ is upwards closed.

	Conversely, define an indexing function given by $I0=I\star = \{ c_{\bot_1},
	c_{\bot_2}\}$ (where the $c_i$ are constant functions as in the previous
	example). Then $0\in \epsilon \Big( \lambda x'. \bigvee_{p\in Ix'} px'\Big) = X$,
	but there is no choice function $p_- : X\rightarrow (X\rightarrow R)$ for $I$
	such that $0\in \epsilon (\lambda x'. p_{x'}x')$. Hence $\epsilon$ is not
	witnessing.
\end{example}

We now define a notion of rationality for nondeterministic games. A strategy profile is
\textit{rational} precisely when there is some plausible hypothesis about how later
players will behave under which that strategy profile is acceptable. As it sounds, this
notion is closely linked to the properties `witnessing' and `upwards closed.' We will
show that witnessing and upwards closed selection functions compute precisely the plays
of rational strategy profiles. We start by restricting ourselves to the two player case
where an arbitrary game is given by $(q:X\times Y\rightarrow R, \eps:
\J_R^{\pf}(X),\delta:\J_R^{\pf}(Y))$.

\begin{definition}[Rational strategy profile]
	Let $\epsilon : \J_R^{\pf}(X)$, $\delta : \J_R^{\pf}(Y)$, and $q: X\times
	Y\rightarrow R$. A strategy profile $(\sigma_1 : X, \sigma_2 : X\rightarrow Y)$ is
	\textit{rational} for $(q,\eps,\delta)$ if
	\begin{enumerate}
		\item There is a choice function $y(-) : X\rightarrow Y$ where for all
			$x\in X$ it holds that $y(x)\in \delta\big( q(x,-)\big)$ and
			\[
				\sigma_1\in\epsilon\big( \lambda x. q(x,y(x))\big);
			\]
		\item For all $x\in X$
			\[
				\sigma_2x\in\delta(q(x,-)).
			\]
	\end{enumerate}
	The set of \textit{rational plays} is 
	\[\rat(q,\epsilon,\delta) = \Big\{ (x,y)\in X\times Y \Bigm| (x,y) = (\sigma_1,\sigma_2\sigma_1) \text{
	for some rational } (\sigma_1,\sigma_2)\Big\}.\]
\end{definition}

Rational strategy profiles and witnessing/upwards closed selection functions are related
by the following theorem.

\begin{theorem}\label{BIGTHEOREM}
	Let $\epsilon : \J_R^{\pf}(X)$ be a multi-valued selection function. The following
	equivalences hold.
	\begin{enumerate}
		\item $\epsilon$ is witnessing if and only if for any $q:X\times
			Y\rightarrow R$ and $\delta : \J_R^{\pf}(Y)$ it holds that
			$(\epsilon\otimes\delta)(q)\subseteq \rat(q,\epsilon,\delta)$.
		\item $\epsilon$ is upwards closed if and only if for any $q:X\times
			Y\rightarrow R$ and $\delta :\J_R^{\pf}(Y)$ it holds that $\rat
			(q,\epsilon,\delta)\subseteq (\epsilon\otimes\delta)(q)$.
	\end{enumerate}
\end{theorem}

\begin{proof} 
  
	We first prove the forward directions of both equivalences which follow by
	definition chasing.

	\underline{1. $\Rightarrow:$} Suppose that $\epsilon$ is witnessing and
	$(x,y)\in(\epsilon\otimes\delta)(q)$. That is,
	\[
		x\in\epsilon\bigg(\lambda x'. \bigvee_{y'\in\delta(q(x',-))}q(x',y') \bigg)
	\]
	and
	\[
		y \in \delta \big( q(x,-)\big).
	\]
	We think of $\delta(q(-,-))$ as an indexing function $I : X\rightarrow \pf
	(X\rightarrow R)$ given by
	\[
		x' \mapsto \Big\{ q(-,y') \Bigm| y'\in\delta(q(x',-)) \Big\}.
	\]
	Then, as $\epsilon$ is witnessing, there is a choice function $y(-) : X\rightarrow
	Y$ where for all $x'$ we have that $y(x')\in\delta(q(x',-))$ such that
	\[
		x\in \epsilon \big( \lambda x'. q(x',y(x'))\big).
	\]
	Then a strategy profile $(\sigma_1,\sigma_2)$ where $\sigma_1 = x$ and 
	\[
		\sigma_2x' = \begin{cases}
				y &\text{if } x'=x
				\\
				y'\in\delta(q(x',-)) &\text{otherwise}
		\end{cases}
	\]
	is rational with play $(x,y)$.

	\underline{2. $\Rightarrow:$} Suppose that $\epsilon$ is upwards closed and
	$(\sigma_1,\sigma_2)$ is rational. Then for all $x\in X$ it holds that
	$\sigma_2x\in\delta\big(q(x,-)\big)$. In particular, $\sigma_2\sigma_1 \in
	\delta\big( q(\sigma_1,-)\big)$. Also there is $y(-) :X\rightarrow Y$ where, for
	all $x$ in $X$, $y(x)\in\delta\big(q(x,-)\big)$ and $\sigma_1\in\epsilon\big(\lambda
	x. q(x,y(x))\big)$. As $\epsilon$ is upwards closed,
	\[
		\sigma_1 \in \epsilon\bigg(\lambda x. \bigvee_{y\in\delta(q(x,-))}q(x,y)
		\bigg).
	\]
	Hence $(\sigma_1,\sigma_2\sigma_1)\in(\epsilon\otimes\delta)(q)$.

	For the backwards directions of the two equivalences we will construct a
	pathological counter example and prove the contrapositive. Define an outcome
	function $q: X\times (X\rightarrow R) \rightarrow R$ to be function application.
	That is, $q(x,p) = px$. Given an indexing function $I: X\rightarrow \pf
	(X\rightarrow R)$ define $\delta_I : \J_R^{\pf}(X\rightarrow R)$ by
	\[
		\delta_I(p) = \begin{cases}
			Ix' & p = q(x',-)
			\\
			\text{arbitrary} &\text{otherwise.}
		\end{cases}
	\]
	Note that $q(x,-)=q(x',-)$ if and only if $x=x'$ or $|R|<2$. In the latter case
	the theorem
	holds vacuously for $R=\varnothing$ and, for $|R|=1$, we have that $|X\rightarrow
	R| = 1$ and so $Ix = Ix'$ for all $x,x'\in X$. Consequently, $\delta_I$ is
	well-defined.

	\underline{1. $\Leftarrow:$} Suppose $\epsilon$ is not witnessing. Then there is
	some indexing function $I : X\rightarrow \pf (X\rightarrow R)$ and
	\[
		x\in \epsilon \bigg( \lambda x'.\bigvee_{p\in Ix'}px'\bigg)
	\]
	such that there is no choice function $p_- : X\rightarrow (X\rightarrow R)$ for
	$I$ such that $x\in\epsilon(\lambda x'. p_{x'}x' )$. By construction,
	\[
		\lambda x'. \bigvee_{p\in Ix'} px' = \lambda x'.
		\bigvee_{p\in\delta_I(q(x',-))} q(x,p)
	\]
	Then $(x,p)\in(\epsilon\otimes\delta_I)(p)$ for any $p\in\delta_I(q(x,-))$. By
	hypothesis there is no choice function $p_- : X\rightarrow (X\rightarrow R)$ for
	$I$ such that $x\in\epsilon(\lambda x'. p_{x'}x')$ and hence there are no rational
	strategy profiles with play $(x,p)$. Hence
	$(\epsilon\otimes\delta_I)(p)\not\subseteq \rat(q,\epsilon,\delta_I)$.

	\underline{2. $\Leftarrow:$} Suppose $\epsilon$ is not upwards closed. Then there
	is $I :X\rightarrow \pf(X\rightarrow R)$ and a choice function $p_- : X\rightarrow
	(X\rightarrow R)$ for $I$ such that $x\in\epsilon (\lambda x'. p_{x'}x')$ and
	\[
	x\not\in \epsilon\bigg( \lambda x'. \bigvee_{p\in Ix'} px'\bigg) =
	\epsilon\bigg(\lambda x'.\bigvee_{p\in\delta_I(q(x',-))} q(x',p)\bigg)
	\]
	Define $\sigma_2 : X\rightarrow (X\rightarrow R)$ by $\sigma_2 (x') = p_{x'}$.
	Then $(x,\sigma_2)$ is rational but
	$(x,\sigma_2x)\not\in(\epsilon\otimes\delta_I)(q)$.
\end{proof}

This theorem has an easy corollary regarding selection functions which are both witnessing
and upwards closed. 

\begin{corollary}
Suppose $\epsilon : \J_R^{\pf}(X)$ is witnessing and upwards closed. Then for all $q:
X\times Y\rightarrow R$ and all $\delta : \J_R^{\pf}(Y)$, $(\epsilon\otimes\delta)(q) =
\rat (q,\epsilon,\delta)$.
\end{corollary}

The property of being witnessing is not closed under the independent product of selection
functions. In section \ref{dominated-strategies} we will see an example where
$\epsilon$ and $\delta$ are both witnessing and upwards closed, but where
$(\epsilon\otimes\delta)$ is not witnessing. A heuristic for why witnessing fails is that
it might be possible to choose witnesses for $\epsilon$ and $\delta$, but be impossible to
choose such witnesses simultaneously. The property of being upwards closed \emph{is}
closed under the independent product of selection functions.

\begin{proposition}\label{proof3}
	Suppose $\eps : \J_R^{\pf}(X)$ and $\delta : \J_R^{\pf}(Y)$ are upwards closed.
	Then $(\eps\otimes\delta)$ is upwards closed.
\end{proposition}

\section{Relation to Subgame Perfect Nash Equilibria}

A standard solution concept for games with sequential play is the
\textit{subgame perfect (Nash) equilibrium} (SPE). One may conjecture that
selection functions can be chosen such that their product computes the
\emph{set of all} plays of subgame perfect Nash equilibria. We show that this conjecture is false: It is impossible to compute the set of all SPE plays using the product of selection functions.

Subgame perfect Nash
equilibria are defined as follows for higher-order sequential games.

\begin{definition}
Consider a 2-player higher order sequential game with outcome function $q : X \times Y \to R$ and selection functions $\epsilon : \J^{\P}_R X$ and $\delta : \J^{\P}_R Y$.
A strategy profile $(\sigma_1,\sigma_2)$ is called a \emph{subgame perfect Nash equilibrium} if the following two conditions hold:
\begin{itemize}
\item $\sigma_1\in \epsilon (\lambda x . q(x,\sigma_2 (x)))$, and
\item $\sigma_2 (x) \in \delta(\lambda y . q(x,y))$ for all $x \in X$
\end{itemize}
	The set of \textit{subgame perfect plays} of $(q,\epsilon,\delta)$ is
	\[
		\sp (q,\epsilon,\delta) = \Big\{ (\sigma_1,\sigma_2\sigma_1)\in X\times Y
		\:\Bigm|\: (\sigma_1,\sigma_2) \text{ is subgame perfect}\Big\}.
	\]
\end{definition}

For comparison, the definition of an ordinary Nash equilibrium is obtained by weakening the second condition to only be required for $x = \sigma_1$.
In a Nash equilibrium the second player is only required to play optimally \emph{on the equilibrium path}, and a Nash equilibrium is subgame perfect if the second player additionally plays optimally if the first player deviates from equilibrium.

\begin{remark}
	Note that subgame perfect strategy profiles are rational: subgame perfect strategy
	profiles are those where $\eps$'s plausible hypothesis regarding $\delta$'s future
	behaviour is correct.
\end{remark}

We now show that the only selection functions which compute the set of subgame perfect
plays are those which are indifferent between all contexts, representing players who have
no preferences over the outcome of any game.

\begin{theorem}\label{proof4}
	Let $\epsilon : \J_R^{\pf}(X)$. If, for all sets $Y$, selection functions $\delta
	: \J_R^{\pf}(Y)$, and functions $q:X\times Y\rightarrow R$ it holds that
	$(\epsilon\otimes\delta)(q) = \sp (q,\epsilon,\delta)$, then $\epsilon$ is
	constant.
\end{theorem}

\begin{proof} 
  
	The proof proceeds by contradiction. Suppose that for all $\delta$ and $q$,
	$(\epsilon\otimes\delta)(q) = \sp (q,\epsilon,\delta)$, and that $\epsilon$ is not
	constant. That
	is, there exist $k_1,k_2 : X\rightarrow R$ and $x\in X$ such that $x\in \eps(p_1)$
	and $x\not\in \eps(p_2)$. Let $q: X\times (X\rightarrow R)\rightarrow R$ be the
	function application operator, $(x,p)\mapsto px$. Define $\delta :
	\J_R^{\pf}(X\rightarrow R)$ by
	\[
		\delta (p) = \begin{cases}
				\{p_1,p_2\} &p=q(x,-)
				\\
				\{p_1\} &\text{otherwise.}
		\end{cases}
	\]
	As $p_1\ne p_2$, we have that $|R|>1$. Consequently, $q(x,-) = q(x',-)$ if and
	only if $x=x'$. Moreover, $x'\ne x$ implies that $\delta( q(x',-)) = \{ p_1 \}$.
	Consider the play $(x,p_2)$ of $(q,\epsilon,\delta)$ noting that, by construction,
	$(x,p_2)$ is not the play of any subgame perfect strategy profile. Define $p_- :
	X\rightarrow (X\rightarrow R)$ to be the constant mapping $p_{x'}=p_1$ so that
	$x\in\epsilon(\lambda x'.\: p_{x'}x')=\epsilon(p_1)$.

	As all subgame perfect plays are rational, we have that $\epsilon$ is upwards
	closed by \ref{BIGTHEOREM}. Hence $(x,p_2) \in (\eps\otimes\delta)(q)$, but
	we have already established that $(x,p_2)$ is not a subgame perfect play.

\end{proof}

 This proof emphasizes the point that multi-valued selection functions
 fail to compute subgame perfect plays due to the possibility of \emph{indifference} in
 sequential games. In the case where $x$ is played, $\delta$ is indifferent between
 playing $p_1$ or $p_2$ whilst $\eps$ is not. In games where there is no such conflicting
 indifference, witnessing and upwards closed selection functions \emph{do} compute the set
 of subgame perfect plays.

\begin{definition}
	Let $\epsilon : \J_R^{\pf}(X)$, $\delta :\J_R^{\pf}(Y)$, and $q: X\times
	Y\rightarrow R$. We say that $(q,\epsilon,\delta)$ has \textit{coinciding
	indifference} if, for all $x\in X$ and $y,y'\in Y$,
	\[
		y,y'\in\delta\big(q(x,-)\big) \Longrightarrow \eps \big( q(-,y)\big) =
		\eps\big( q(-,y')\big)
	\]
\end{definition}

\begin{proposition}\label{proof5}
	Suppose $(q,\eps,\delta)$ has coinciding indifference and that $\eps$ is
	witnessing and upwards closed. Then $(\eps\otimes\delta)(q)=\sp (q,\eps,\delta)
	=\rat (q,\eps,\delta)$.
\end{proposition}

To summarize, in a two round game with first player $\eps$, $(\eps\otimes\delta)(q)$ computes subgame
perfect plays for arbitrary second player $\delta$ and arbitrary outcome function $q$ if
and only if
$\eps$ is constant. $(\eps\otimes\delta)(q)$ \emph{does} compute subgame perfect plays in
the special cases where $\eps$ is upwards closed and witnessing, and $(q,\eps,\delta)$ has
coinciding indifference.

\section{Relation to strictly dominated strategies}\label{dominated-strategies}

In this section we extend the previous results to games of
arbitrary length. 
Note that in this paper we will only consider games whose length is precisely $n$ (i.e. all plays have length $n$), which includes via an encoding games whose length is bounded by $n$.
We do not consider unbounded games, i.e. games whose plays are all finite but which have arbitrarily long plays, which would introduce significant complications.

\begin{notation}
Given $A \subseteq
\bigcup_{i=1}^n X_i$, we use $A^{(j)}$ to denote $X_j \cap A$. 
\end{notation}

In particular, if $\Gamma$ is a set of strategies for
some sequential game, then $\Gamma^{(j)}$ denotes the strategies in
$\Gamma$ which are strategies at round $j$.

In the 2 player case, if a strategy profile $(\sigma_1,\sigma_2)$ is rational for
$(q,\epsilon,\delta)$, there is some choice function $y(-) :X\rightarrow Y$ such that
$\epsilon$ makes an acceptable play if $\delta$ plays according to $y(-)$. Equivalently,
$\sigma_1$ is rational if there is \emph{some} rational $\sigma_2$ for $\delta$ under
which $\sigma_1$ is a good move. To generalise to the $n$-round case, we can simply extend
this heuristic. Given a game $\big( q, (\epsilon_i : \J_R^{\pf}(X_i))_{i=1}^n\big)$, a
strategy $\sigma_1 : X_1$ is rational if there are strategies $\sigma_2,\cdots,\sigma_n$,
rational for $\epsilon_2,\cdots,\epsilon_n$ respectively, under which $\sigma_1$ is a good
move. For players $\eps_i$ acting in the `mid-game,'  a strategy is rational if it is
rational for all subgames which are given by partial plays $x\in\prod_{j=1}^{i-1}X_j$.

We will define a more general notion of sets of strategies as \textit{consistent} for a
game $\G$. The set of rational strategy profiles will then be realised as the maximal
consistent set of strategy profiles.

\begin{definition}
Let $\Gamma$ be a set of strategies for a sequential game $\G$.
$\Gamma$ is \emph{$\G$-consistent} if
for all $i<n$ and $\sigma_i\in\Gamma^{(i)}$, and all partial plays
$x\in\prod_{j=1}^{i-1} X_j$, there exists $\sigma =
(\sigma_{i+1},\cdots,\sigma_n)$ where  $\sigma_{i+1},\cdots, \sigma_n \in \Gamma$ such that
\[
\sigma_i (x) \in \eps_i \bigg(\lambda y. q\big(
(x,y)^{\sigma}  \big) \bigg)
\]
where $(x,y)^{\sigma}$ is the strategic extension of
$(x_1,\cdots, x_{i-1},y)$ by $(\sigma_{i+1},\cdots,\sigma_n)$.
\end{definition}

Note that if $\Gamma$ is $\G$-consistent, the
$\G$-consistency of $\Gamma \cup \{\sigma_i\}$ depends only on
$\Gamma^{(j)}$ for $j>i$. With that in mind, we can define the maximal
$\G$-consistent set of strategies, denoted by
$\Sigma(\G)$, as follows.

\begin{definition}\label{sigmaG}
$\Sigma(\G)$ is given by
\begin{align*}
\Sigma(\G)^{(n)} &= \Big\{ \sigma_n\in\Sigma^{(n)} : \forall
x\in\prod_{i=1}^{n-1} X_i. \: \sigma_n
(x) \in \epsilon_n \big( q(x,-)\big)\Big\}
\\
\Sigma(\G)^{(i)} &= \Big\{ \sigma_i\in\Sigma^{(i)} : \{\sigma_i\} \cup
\bigcup_{j>i} \Sigma(\G)^{(j)} \mbox{ is }
\G\mbox{-consistent} \Big\}.
\end{align*}
\end{definition}

\begin{definition}
Let $\Gamma$ be a set of strategies for a sequential game $\G$. A play
$x\in \prod_{i=1}^n X_i$ is a \emph{$\Gamma$ play} if $x$ is the
strategic play of a strategy profile $\sigma$ where $\sigma_i \in
\Gamma^{(i)}$ for each $i\leq n$.
\end{definition}
The following lemma and theorem provide a generalisation of \ref{BIGTHEOREM} to the
$n$-round
case.

\begin{lemma}\label{mOutcome}
	Let $q: X \rightarrow R$ and define $q' : X\times Y \rightarrow R$ by $q'(x,y) =
	q(x)$. Then, for all $\eps : \J_R^{\pf}(X)$ and $\delta : \J_R^{\pf}(Y)$,
	\[
		x \in \eps(q) \Leftrightarrow \exists y\in Y \text{ such that }
		(x,y)\in(\eps\otimes\delta)(q').
	\]
\end{lemma}

\begin{proof}
	If $x\in\eps(q)$ then, for all $y\in Y$,
	\[
		x\in \eps(q) = \eps\Big( \lambda x'. \bigvee_{y\in\delta(q(x'))}q(x')\Big)
		= \eps\Big(\lambda x'. \bigvee_{y\in\delta(q'(x',-))}q'(x',y)\Big).
	\]
	Hence if $y\in \delta\big(q(x,-)\big)$ then $(x,y)\in(\eps\otimes\delta)(q')$.

	Conversely,
	\[
		(x,y)\in (\eps\otimes\delta)(q')\Rightarrow x\in\eps\Big(\lambda x'.
		\bigvee_{y'\in\delta(q'(x',-))}q'(x',y') \Big) = \eps(q)
	\]
\end{proof}

\begin{corollary}\label{muteOutcome}
	Let $\Big( q, \big( \eps_i \big)_{i=1}^n \Big)$ be a sequential game. Suppose there
	exists $j<n$ and $q_j : X_i\times X_n \rightarrow R$ such that, for all
	$x\in\prod_{i=1}^n X_i$, $qx = q_j(x_j,x_n)$.
	Then $(x_j,x_n)\in(\eps_j\otimes\eps_n)(q_j)$ if and only if there exist
	$x_1,\cdots,x_{j-1},x_{j+1},\cdots,x_{n-1}$ with each $x_j\in X_j$ such that
	$(x_1,\cdots,x_n)\in \bigg(\bigotimes_{i=1}^n\eps_i\bigg)(q)$.
\end{corollary}

\begin{theorem}\label{nRound}
	Let $\eps_i : \J_R^{\pf}(X_i)$ for $i<n$. For all sets $X_n$, selection functions
	$\eps_n:\J^{\pf}_R(X_n)$, and outcome functions $q: \prod_{i=1}^n X_i \rightarrow
	R$, the following equivalences hold.
\begin{enumerate}
\item  $\epsilon_i$ is witnessing for each $i<n$ if and only if
$\bigg( \bigotimes_{i=1}^n \epsilon_i \bigg) (q)$ is a subset of the set of
$\Sigma(\G)$ plays.
\item   $\epsilon_i$ is upwards closed for each $i<n$ if and only if $\bigg(
\bigotimes_{i=1}^n \epsilon_i \bigg)(q)$ is a superset of the set of
$\Sigma(\G)$ plays.
\end{enumerate}
\end{theorem}

\begin{proof}
	We prove the forward directions of the two equivalences first. The proof proceeds
	by induction on $n$, noting that the cases $n=1$ are trivial.

	$(1):$ Suppose \[x=(x_1,\cdots,x_n)\in\bigg(\bigotimes_{i=1}^n \epsilon_i \bigg)
	(q)\]. As
$\epsilon_1$ is witnessing, it is the play of some rational
strategy profile $(x_1,
f:X_1\rightarrow \prod_{i=2}^n X_i)$ of the two round game $\big( (X_1,
\prod_{i=2}^n X_i), (\eps_1, \bigotimes_{i=2}^n \eps_i), q
\big)$. By hypothesis we have that 
\[ \bigg(
\bigotimes_{i=2}^n \epsilon_i \bigg)(q(y_1,-)) \]
 is a subset of the set of
$\Sigma(\G^{y_1})$ plays for all $y_1\in X_1$ where 
\[\G^{y_1}= \big(
(X_i)_{i=2}^n, (\eps_i)_{i=2}^n, q(y_1,-)\big)\]
 Hence
$f(y_1)$ is the play of some $\Sigma(\G^{y_1})$-consistent
strategy profile $\sigma^{y_1}$ for all $y_1\in X_1$. Then the
strategy profile $\tau$ for $\G$ given by
\begin{align*}
\tau_1 &= x_1
\\
\tau_{i+1}(y_1,\cdots,y_i) &= \sigma^{y_1}_{i+1}(y_2,\cdots,y_i)
\end{align*}
is such that $\tau_i\in\Sigma(\G)$ for all $i$ and the play of $\tau$
is $x$.

$(2):$ Suppose that $x=(x_1,\cdots,x_n)$ is the $\Sigma(\G)$ play of
$(\sigma_1,\cdots,\sigma_n)$. A simple check demonstrates that for all $y_1\in X_1$, we
have that $\sigma_2,\cdots,\sigma_n\in\Sigma(\G^{y_1})$. By hypothesis, the strategic play
$y_1^{\sigma_{-1}}$ of $(\sigma_2,\cdots,\sigma_n)$ for the game $\G^{y_1}$ is such that
\[
	y_1^{\sigma_{-1}} \in \bigg( \bigotimes_{i=2}^n\eps_i\bigg)
\big(q(y_1,-)\big).
\]
In particular, 
\begin{equation}\tag{$\star$}
	(x_2,\cdots,x_n)\in \bigg(\bigotimes_{i=2}^n\eps_i\bigg)\big( q(x_1,-)\big).
\end{equation}

As $x_1=\sigma_1\in\Sigma(\G)$ there exists $\tau = (\tau_2,\cdots,\tau_n)$ with each
$\tau_i\in\Sigma(\G)$ such that 
\[
	x_1 \in \eps_1 \big( \lambda y_1. q(y_1^{\tau})\big)
\]
and, for all $y_1\in X_1$,
\[
	(y_1^\tau)_{-1}\in \bigg(\bigotimes_{i=2}^n\eps_i\bigg) \big(q(y_1,-)\big).
\]
As $\eps_1$ is upwards closed,
\[
	x_1\in \eps_1\bigg( \lambda (y_1). \bigvee_{z\in Ay_1} q(x_1,z)\bigg)
\]
where $A(y_1) = \Big( \bigotimes_{i=2}^n \eps_i \Big)\big(q(y_1,-)\big)$. From this and
$(\star)$ we conclude that
\[
	x\in\bigg(\bigotimes_{i=1}^n\eps_i\bigg)(q).
\]

As for the backward directions, for $i<n$ consider the construction
$\delta^i_I:\J_R^{\pf}(X_i)$ as in the proof of
\ref{BIGTHEOREM} and let $q_i : \Big(\prod_{j=1}^{n-1} X_j\Big)\times (X_i\rightarrow
R)\rightarrow R$ be given by $(x,p)\mapsto px_i$. The converse directions are then a
corollary of \ref{muteOutcome} and \ref{BIGTHEOREM} by considering the game $(q,
(\eps_1,\cdots,\eps_{n-1},\delta_I))$ for each $i$.

\end{proof}

We have seen that nondeterministic selection functions do not, in
general, describe subgame perfect Nash equilibria. We have also given
a technical characterisation of the plays nondeterministic selection
functions \emph{do} describe. In this section we make sense of this
technical characterisation, relating it to a solution concept that is
already well-known.

\begin{definition}
	Let $S,T\subseteq \R$. $S$ \textit{strictly dominates} $T$ if $\min (S) > \max
	(T)$. We write $S \succ_s T$.
\end{definition}

We now define the \textit{strict dominance selection functions} to be those that return
the set of choices that are not mapped to strictly dominated subsets of the reals for a
given context.

\begin{definition}
	Let $R$ be $\pf (\R^n)$ where the semilattice join is given by union
	(equivalently, the order structure is given by inclusion). Given $p :
	X_i\rightarrow \pf(\R^n)$, define $p^i : X_i\rightarrow\pf(\R)$ to be
	$(\pf\pi_i)\circ p$. Define the
	$i^{\text{th}}$ \textit{strict dominance} selection function, $\eps_i^s :
	\J_R^{\pf}(X_i)$ by
	\[
		\eps_i^s(p: X_i\rightarrow \pf(\R^n)) = \bigg\{ x_i\in X_i \biggm| \forall
			x_i'\in X_i,\:\: p^ix_i
		\not\prec_sp^ix_i'\bigg\}.
	\]
\end{definition}

The strict dominance selection functions are witnessing and upwards closed, demonstrating
that they provide an appropriate solution concept for multi-valued selection functions.

\begin{proposition}\label{proof7}
	$\eps_i^s$ is witnessing and upwards closed. 
\end{proposition}

The product of strict dominance selection functions provides an example of when the
product of two witnessing selection functions is not witnessing.

\begin{proposition}\label{proof8}
	Let $X=\{0,1\}$ and let $\eps = \eps_1^s$ and $\delta = \eps_2^s$ (so $\eps,\delta
	: \J_R^{\pf} (X)$ for $R = \pf (\R^2)$). Then $(\eps\otimes\delta)$ is not
	witnessing.
\end{proposition}

Consider the game given by $\G = \big( q, (\eps_1^s,\cdots,\eps_n^s)\big)$. By
\ref{nRound}, we know that $\big( \bigotimes_{i=1}^n \eps_i \big)(q)$ is equal to the set
of $\Sigma(\G)$ plays. The set of strategies $\Sigma(\G)$ is then the maximal set of
strategies such that no strategy is strictly dominated in any subgame. When each $X_i$ is
finite, this is the same as $\big( \bigotimes_{i=1}^n \eps_i^s\big)(q)$ computing the
plays of strategies obtained via the iterated removal of strictly dominated
strategies. This iterated removal of strictly dominated strategies is a well known
solution concept \cite{fudenberg1991game,Mas-Colell1995,osborne1994course} and goes back to \cite{gale1953theory} and
\cite{luce2012games}. In economic game theory this concept is typically applied statically
to normal form games whereas in the context of selection functions we apply it
to sequential games. 

\bibliography{bib.bib}

\section{Appendix: Proofs}

\begin{proof}[Proof of Proposition \ref{proof1}]
  
	Suppose $R$ is total and $I : X\rightarrow \pf (X\rightarrow R)$. Then for all
	$x'\in X$ there exists $p_{x'}\in Ix'$ such that
	\[
		\bigvee_{p\in Ix'}px' = p_{x'}x'.
	\]
	Then
	\[
		x\in\epsilon\bigg( \lambda x'.\bigvee_{p\in Ix'}px'\bigg) \Longrightarrow
		x\in\epsilon\big(\lambda x'. p_{x'}x'\big)
	\]
\end{proof}

\begin{proof}[Proof of Proposition \ref{proof3}]
	Suppose $(x,y) \in (\eps\otimes\delta)(\lambda (x',y') p_{(x',y')}(x',y'))$ where
	$p_{(-,-)} : X\times Y\rightarrow (X\times Y\rightarrow R)$ is a choice function
	for some $I : X\times Y\rightarrow \pf(X\times Y\rightarrow R)$. Then
	\begin{enumerate}
		\item $y\in \delta(\lambda y'. p_{(x,y')}(x,y'))$; and
		\item \[
				x\in \eps\Big( \lambda x'.
					\bigvee_{y'\in\delta(p_{(x',-)}(x',-))}
				p_{(x',y')}(x',y')\Big).
				\]
	\end{enumerate}
	In order to show that $(x,y)\in (\eps\otimes\delta)\big(\lambda (x',y').
	\bigvee_{p\in I(x',y')}p(x',y')\big)$ we need to show 
	\begin{enumerate}[label=(\alph*)]
		\item  \[
				y\in \delta\Big( \lambda y'. \bigvee_{p\in
				I(x,y')}p(x,y')\Big)
			\]
			and
		\item \[
				x\in \eps\Big(\lambda x'. \bigvee_{y'\in Ax'}
				\bigvee_{p\in I(x',y')} p(x',y') \Big).
			\]
			where $Ax' = \delta \Big(\lambda y''. \bigvee_{p\in
			I(x',y'')}p(x',y'') \Big)$.
	\end{enumerate}
	As $\delta$ is upwards closed we have that, for all $x'\in X$,
	\[
		y'\in \delta \Big( \lambda y''. p_{(x',y'')}(x',y'')\Big) \Rightarrow
		y'\in Ax'.
	\]
	In particular, (a) holds. By using upwards closure of $\eps$ twice, we have
	\begin{align*}
		x\in \eps\Big( \lambda x'.
					\bigvee_{y'\in\delta(p_{(x',-)}(x',-))}
				p_{(x',y')}(x',y')\Big)
		&\Rightarrow x\in \Big( \lambda x'. \bigvee_{y'\in Ax'}
		p_{(x',y')}(x',y')\Big)
		\\
		&\Rightarrow x\in \eps\Big( \lambda x'. \bigvee_{y'\in Ax'} \bigvee_{p\in
		I(x',y')}p(x',y')\Big)
	\end{align*}
\end{proof}

\begin{proof}[Proof of Proposition \ref{proof5}]
	Let $(x,y)\in(\eps\otimes\delta)(q)$. By \ref{BIGTHEOREM}, $(x,y)$ is the
	play of some rational strategy profile $(\sigma_1,\sigma_2)$. Then there exists
	some function $y(-):X\rightarrow Y$ where, for all $x'\in X$, $y(x')\in\delta\big(
	q(x',-)\big)$ and $\sigma_1 \in\eps(\lambda x'.\: q(x',y(x'))\big)$. By coinciding
	indifference, $\sigma_1 \in\eps\big(\lambda x'. q(x',\sigma_2 x')\big)$.

	Conversely, subgame perfect plays are rational. Hence, if $(x,y)$ is a subgame
	perfect play, then $(x,y)\in(\eps\otimes\delta)(q)$ by \ref{BIGTHEOREM}.
\end{proof}

\begin{proof}[Proof of corollary \ref{muteOutcome}]
	The proof proceeds by a routine induction on $n$, noting that the case $n=2$ is
	trivial. When $j\ne 1$ the result follows easily by choosing
	\[
		x_1\in\eps_1\Big(\lambda y_1. \bigvee_{(y_2,\cdots, y_n)\in
		A(y_2,\cdots,y_n)} q(y_1,\cdots,y_n) \Big)
	\]
	where $A(y_1) = (\bigotimes_{i=2}^m \eps_i)(q(y_1,-))$ and applying the inductive
	hypothesis to the game $(q(x_1), \big( \eps_i \big)_{i=2}^n )$. For the case
	$j=1$, note that the function
	\[
		\lambda y_1,\cdots, y_n \bigvee_{y_n\in\eps_n( q(y_1,\cdots, y_{n-1},-))}
		q(y_1,\cdots,y_n)
	\]
	is mute in every variable except $y_1$. Then, using \ref{mOutcome},
	\begin{align*}
		&(x_1,x_n)\in (\eps_1\otimes\eps_n)(q_1)
		\\
		&\Leftrightarrow x_1\in\eps_1\Big(\lambda y_1.
		\bigvee_{y_n\in\eps_n(q_1(y_1,-))}q_1(y_1,y_n)\Big) \text{ and } x_n \in
		\eps_n \big(q_1 (x_1,-)\big)
		\\&\\
		&\Leftrightarrow \exists x\in\prod_{i=2}^{n-1}X_i. \:\: (x_1,x)\in
		\Big( \bigotimes_{i=1}^{n-1} \eps_i \Big)\Big(\lambda
		y_1,y.\bigvee_{y_n\in\eps_n(q(y_1,y,-))} q(y_1,y,y_n) \Big)
	      \\
		&\hphantom{\Leftrightarrow} 
		\text{ and } x_n \in \eps_n \big( q(x_1,x,-)\big)
	      \\&\\
		&\Leftrightarrow \exists x\in \prod_{i=1}^{n-1}X_i. \:\:
		(x_1,x,x_n)\in\Big(\bigotimes_{i=1}^n \eps_i \Big)(q).
	\end{align*}
\end{proof}

\begin{proof}[Proof of Proposition \ref{nRound}]
	We prove the forward directions of the two equivalences first. The proof proceeds
	by induction on $n$, noting that the cases $n=1$ are trivial.

	$(1):$ Suppose \[x=(x_1,\cdots,x_n)\in\bigg(\bigotimes_{i=1}^n \epsilon_i \bigg)
	(q)\]. As
$\epsilon_1$ is witnessing, it is the play of some rational
strategy profile $(x_1,
f:X_1\rightarrow \prod_{i=2}^n X_i)$ of the two round game $\big( (X_1,
\prod_{i=2}^n X_i), (\eps_1, \bigotimes_{i=2}^n \eps_i), q
\big)$. By hypothesis we have that 
\[ \bigg(
\bigotimes_{i=2}^n \epsilon_i \bigg)(q(y_1,-)) \]
 is a subset of the set of
$\Sigma(\G^{y_1})$ plays for all $y_1\in X_1$ where 
\[\G^{y_1}= \big(
(X_i)_{i=2}^n, (\eps_i)_{i=2}^n, q(y_1,-)\big)\]
 Hence
$f(y_1)$ is the play of some $\Sigma(\G^{y_1})$-consistent
strategy profile $\sigma^{y_1}$ for all $y_1\in X_1$. Then the
strategy profile $\tau$ for $\G$ given by
\begin{align*}
\tau_1 &= x_1
\\
\tau_{i+1}(y_1,\cdots,y_i) &= \sigma^{y_1}_{i+1}(y_2,\cdots,y_i)
\end{align*}
is such that $\tau_i\in\Sigma(\G)$ for all $i$ and the play of $\tau$
is $x$.

$(2):$ Suppose that $x=(x_1,\cdots,x_n)$ is the $\Sigma(\G)$ play of
$(\sigma_1,\cdots,\sigma_n)$. A simple check demonstrates that for all $y_1\in X_1$, we
have that $\sigma_2,\cdots,\sigma_n\in\Sigma(\G^{y_1})$. By hypothesis, the strategic play
$y_1^{\sigma_{-1}}$ of $(\sigma_2,\cdots,\sigma_n)$ for the game $\G^{y_1}$ is such that
\[
	y_1^{\sigma_{-1}} \in \bigg( \bigotimes_{i=2}^n\eps_i\bigg)
\big(q(y_1,-)\big).
\]
In particular, 
\begin{equation}\tag{$\star$}
	(x_2,\cdots,x_n)\in \bigg(\bigotimes_{i=2}^n\eps_i\bigg)\big( q(x_1,-)\big).
\end{equation}

As $x_1=\sigma_1\in\Sigma(\G)$ there exists $\tau = (\tau_2,\cdots,\tau_n)$ with each
$\tau_i\in\Sigma(\G)$ such that 
\[
	x_1 \in \eps_1 \big( \lambda y_1. q(y_1^{\tau})\big)
\]
and, for all $y_1\in X_1$,
\[
	(y_1^\tau)_{-1}\in \bigg(\bigotimes_{i=2}^n\eps_i\bigg) \big(q(y_1,-)\big).
\]
As $\eps_1$ is upwards closed,
\[
	x_1\in \eps_1\bigg( \lambda (y_1). \bigvee_{z\in Ay_1} q(x_1,z)\bigg)
\]
where $A(y_1) = \Big( \bigotimes_{i=2}^n \eps_i \Big)\big(q(y_1,-)\big)$. From this and
$(\star)$ we conclude that
\[
	x\in\bigg(\bigotimes_{i=1}^n\eps_i\bigg)(q).
\]

As for the backward directions, for $i<n$ consider the construction
$\delta^i_I:\J_R^{\pf}(X_i)$ as in the proof of
\ref{BIGTHEOREM} and let $q_i : \Big(\prod_{j=1}^{n-1} X_j\Big)\times (X_i\rightarrow
R)\rightarrow R$ be given by $(x,p)\mapsto px_i$. The converse directions are then a
corollary of \ref{muteOutcome} and \ref{BIGTHEOREM} by considering the game $(q,
(\eps_1,\cdots,\eps_{n-1},\delta_I))$ for each $i$.

\end{proof}

\begin{proof}[Proof of Proposition \ref{proof7}]
	Let $I : X_i\rightarrow \pf(X_i\rightarrow \pf(\R^n))$. Suppose first that
	\[x\in\eps_i^s(\lambda x'. \bigcup_{p\in Ix'} px').\]
	That is, for all $x'\in X_i$,
	\[
		\max (\bigcup_{p\in Ix}p^ix )  \geq \min (\bigcup_{p\in Ix'}p^i x')
	\]
	Then, setting $p_x\in Ix$ to be a function attaining the maximum of $\bigcup_{p\in
	Ix}p^ix$ and, for $x'\ne x$, setting $p_{x'}\in Ix'$ to be a function
	attaining the minimum of $\bigcup_{p\in Ix'}p^i x'$, we define a choice
	function $p_- : X_i\rightarrow (X_i\rightarrow \pf(\R^n))$ such that
	\[
		x \in \eps_i^s (\lambda x'. p_{x'}x').
	\]
	Hence $\eps_i^s$ is witnessing.

	It is similarly easy to show that $\eps_i^s$ is upwards closed as
	\[
		\max (p_{x}^i x ) \geq \min (p_{x'}^i x') \Longrightarrow 
		\max(\bigcup_{p\in Ix} p^i x) \geq \min (\bigcup_{p\in Ix'}p^ix').
	\]
\end{proof}

\begin{proof}[Proof of Proposition \ref{proof8}]
	Define contexts $p_\eps, p_\delta, p_0 : X^2 \rightarrow \pf(\R^2)$ by
	\begin{align*}
		p_\eps (x,x') &=\begin{cases}
			\{(1,-1)\} &x=x'=0 \\
			\{(0,0)\} &\text{otherwise}
		\end{cases}
			    \\&\\
		p_\delta (x,x') &= \begin{cases}
			\{(-1,1)\} &x=x'=0 \\
			\{(0,0)\}  &\text{otherwise}
		\end{cases}
			      \\&\\
			p_0(x,x') &= \{(0,0)\}.
	\end{align*}
	Define $I : X^2\rightarrow \pf (X^2\rightarrow \pf(\R^2))$ by
	\[
		I(x,x') =\begin{cases}
			\{p_\eps, p_\delta \} & x=x'=0
			\\
			\{ p_0 \} &\text{otherwise.}
		\end{cases}
	\]
	We think of $\eps$ and $\delta$ as playing the following game where the outcome
	function is chosen nondeterministically. 
	\begin{center}\begin{tikzpicture}
		\node (e) {$\eps$};
		\path (e) +(1.5,1.5) node (d1) {$\delta$};
		\path (e) +(1.5,-1.5) node (d2) {$\delta$};

		\draw (e.north east) -- node[above,pos=0.4,yshift=3pt]{$0$} (d1.225);
		\draw (e.south east) -- node[below,pos=0.4,yshift=-3pt]{$1$} (d2.135);
		
		\path (d1) +(1,1) node (o1)[anchor=west]{$\{p_\eps,p_\delta\}$};
		\draw (d1.north east) -- node[above,pos=0.4,yshift=3pt]{$0$} (o1.west);

		\path (d1) +(1,-1) node (o2)[anchor=west]{$\{p_0\}$};
		\draw (d1.south east) -- node[below,pos=0.4,yshift=-3pt]{$1$} (o2.west);

		\path (d2) +(1,1) node (o3)[anchor=west]{$\{p_0\}$};
		\draw (d2.north east) -- node[above,pos=0.4,yshift=3pt]{$0$} (o3.west);
		
		\path (d2) +(1,-1) node (o4)[anchor=west]{$\{p_0\}$};
		\draw (d2.south east) -- node[below,pos=0.4,yshift=-3pt]{$1$} (o4.west);
	\end{tikzpicture}\end{center}
	We will see that $(\eps\otimes\delta)$ fails to be witnessing as $\eps$ is
	satisfied with playing $0$ in the case $(0,0)$ results in $p_\eps$ and $\delta$ is
	satisfied playing $0$ when $(0,0)$ results in $p_\delta$, but there is no possible
	resulting context under which both $\eps$ and $\delta$ are happy to choose $0$.
	Indeed, simple checks verify that
	\[
		(0,0) \in (\eps\otimes\delta) \bigg(\lambda (x,y). \bigcup_{p\in I(x,y)}
		p(x,y) \bigg)
	\]
	but that there is no choice function $p_ : X^2\rightarrow
	(X^2\rightarrow\pf(R^2))$ for $I$ with $(0,0) \in (\eps\otimes\delta)(\lambda
	(x,y). p_{(x,y)}(x,y))$.
\end{proof}

\end{document}

%% file: tikzstyles.tikzstyles
\tikzstyle{circ}=[fill=white, draw=black, shape=circle]
\tikzstyle{blank}=[fill=white, draw=white, shape=circle]
\tikzstyle{none}=[fill=none, draw=none]

\tikzstyle{copy}=[fill=white,draw=black,shape=circle,minimum height =
0.2cm,inner sep=0]

\tikzstyle{varCopy}=[fill=black,draw=black,shape=circle,minimum height=0.2cm,inner sep=0]

\tikzstyle{copy2}=[fill=black,draw=black,shape=circle,minimum height =
0.2cm,inner sep=0]

\tikzstyle{1morph1}=[fill=white,draw=black,shape=rectangle,
minimum width=1cm,minimum height=1cm]

\tikzstyle{1morph}=[fill=white,draw=black,shape=rectangle,minimum
width=0.75cm,minimum height=0.75cm,inner sep=0.1cm]

\tikzstyle{2morph2}=[fill=white,draw=black,shape=rectangle,minimum
width = 1cm, minimum height = 2cm]

\tikzstyle{2morph}=[fill=white,draw=black,shape=rectangle,minimum
width=1cm,minimum height=1.25cm,inner sep=0.1cm]

\tikzstyle{nmorph}=[fill=white,draw=black,shape=rectangle,minimum height = 6cm, minimum
width = 1cm, inner sep=0.1cm]

\tikzstyle{1state}=[fill=white,draw=black,regular polygon, regular
polygon sides = 3, minimum height = 0.5cm, regular polygon rotate = -30]

\tikzstyle{dbox}=[fill=white,draw=black,dashed,shape=rectangle,minimum
width=2cm,minimum height=1cm,inner sep=0.1cm]

\tikzstyle{vdbox}=[fill=white,draw=black,dashed,shape=rectangle,minimum
width=2cm,minimum height=1.5cm,inner sep=0.1cm]

\tikzstyle{2state}=[inner sep=0.05cm,fill=white,draw=black,isosceles
triangle, minimum
width=1.25cm,isosceles triangle apex angle = 90,shape border rotate = 180]

\tikzstyle{var2state}=[inner sep=0.05cm,fill=white,draw=black,isosceles
triangle, minimum
width=1.25cm,isosceles triangle apex angle = 60,shape border rotate = 180]

\tikzstyle{g2state}=[inner sep=0.05cm,fill=white,draw=black,isosceles
triangle,minimum width = 6cm, isosceles triangle apex angle = 110,shape border rotate =180]

\tikzstyle{g2effect}=[inner sep=0.05cm,fill=white,draw=black,isosceles
triangle,minimum width = 6cm, isosceles triangle apex angle = 110]

\tikzstyle{2effect}=[inner sep=0.05cm,fill=white,draw=black,isosceles triangle,minimum
width=1.25cm, isosceles triangle apex angle = 90]

\tikzstyle{b2effect}=[inner sep=0.05cm,fill=white, draw=black, isosceles triangle,
minimum width =2cm, isosceles triangle apex angle=90]

\tikzstyle{arrow}=[->]

\tikzset{->-/.style={decoration={markings,mark=at position .5 with {\arrow{>}}},
postaction={decorate}}}

%% file: mselection.bbl
\begin{thebibliography}{10}

\bibitem{escardo_oliva_selection_functions_bar_recursion_backward_induction}
Martin Escard\'o and Paulo Oliva.
\newblock Selection functions, bar recursion and backward induction.
\newblock {\em Mathematical structures in computer science}, 20(2):127--168,
  2010.
\newblock \href {http://dx.doi.org/10.1017/S0960129509990351}
  {\path{doi:10.1017/S0960129509990351}}.

\bibitem{escardo11}
Martin Escard\'o and Paulo Oliva.
\newblock Sequential games and optimal strategies.
\newblock {\em Proceedings of the Royal Society A}, 467:1519--1545, 2011.
\newblock \href {http://dx.doi.org/10.1098/rspa.2010.0471}
  {\path{doi:10.1098/rspa.2010.0471}}.

\bibitem{escardo_oliva_herbrand_interpretation}
Martin Escard\'o and Paulo Oliva.
\newblock The {H}erbrand functional interpretation of the double negation
  shift.
\newblock {\em Journal of symbolic logic}, 82(2):590--607, 2017.
\newblock \href {http://dx.doi.org/10.1017/jsl.2017.8}
  {\path{doi:10.1017/jsl.2017.8}}.

\bibitem{fudenberg1991game}
Drew Fudenberg and Jean Tirole.
\newblock {\em Game theory}.
\newblock MIT Press, 1991.

\bibitem{gale1953theory}
David Gale.
\newblock A theory of n-person games with perfect information.
\newblock {\em Proceedings of the National Academy of Sciences},
  39(6):496--501, 1953.

\bibitem{hedges_monad_transformers_backtracking_search}
Jules Hedges.
\newblock Monad transformers for backtracking search.
\newblock In {\em Proceedings of Mathematically Structured Functional
  Programming (MSFP) 2014}, EPTCS, pages 31--50, 2014.
\newblock \href {http://dx.doi.org/10.4204/EPTCS.153.3}
  {\path{doi:10.4204/EPTCS.153.3}}.

\bibitem{hedges_towards_compositional_game_theory}
Jules Hedges.
\newblock {\em Towards compositional game theory}.
\newblock PhD thesis, Queen Mary University of London, 2016.

\bibitem{hedges_etal_higher_order_decision_theory}
Jules Hedges, Paulo Oliva, Evguenia Shprits, Viktor Winschel, and Philipp Zahn.
\newblock Higher-order decision theory.
\newblock In J\"org Rothe, editor, {\em Algorithmic Decision Theory}, volume
  10576 of {\em Lecture Notes in Artificial Intelligence}, pages 241--254.
  Springer, 2017.
\newblock \href {http://dx.doi.org/10.1007/978-3-319-67504-6_17}
  {\path{doi:10.1007/978-3-319-67504-6_17}}.

\bibitem{hedges_etal_selection_equilibria_higher_order_games}
Jules Hedges, Paulo Oliva, Evguenia Shprits, Viktor Winschel, and Philipp Zahn.
\newblock Selection equilibria of higher-order games.
\newblock In {\em Practical aspects of declaritive languages}, volume 10137 of
  {\em Lecture Notes in Computer Science}, pages 136--151. Springer, 2017.
\newblock \href {http://dx.doi.org/10.1007/978-3-319-51676-9_9}
  {\path{doi:10.1007/978-3-319-51676-9_9}}.

\bibitem{luce2012games}
R~Duncan Luce and Howard Raiffa.
\newblock {\em Games and decisions: Introduction and critical survey}.
\newblock Wiley, 1957.

\bibitem{Mas-Colell1995}
Andreu Mas-Colell, Michael~D. Whinston, and Jerry Green.
\newblock {\em {Microeconomic Theory}}.
\newblock Oxford University Press, 1995.

\bibitem{osborne1994course}
Martin~J Osborne and Ariel Rubinstein.
\newblock {\em A course in game theory}.
\newblock MIT press, 1994.

\bibitem{schwalbe_walker_zermelo_early_history_game_theory}
Ulrich Schwalbe and Paul Walker.
\newblock {Z}ermelo and the early history of game theory.
\newblock {\em Games and economic behaviour}, 34:123--137, 2001.
\newblock \href {http://dx.doi.org/10.1006/game.2000.0794}
  {\path{doi:10.1006/game.2000.0794}}.

\end{thebibliography}
